\documentclass[letterpaper, 10 pt, conference]{ieeeconf}

\IEEEoverridecommandlockouts
\usepackage{graphicx} 
\usepackage{cite}
\usepackage{amsmath}
\usepackage{mathtools}
\usepackage{amssymb}
\usepackage{xcolor}
\usepackage{float}

\newtheorem{problem}{Problem}

\newtheorem{theorem}{Theorem}

\title{\LARGE \bf GMM-Based Time-Varying Coverage Control}
\author{Behzad Zamani\thanks{B. Zamani, A. Chapman, P. Dower and C. Manzie are with the University of Melbourne. J. Kennedy and S. Crase are with the Defence Science and Technology (DST) Group.}\thanks{This work was supported by the Defence Science and
Technology Group, Australia, via Centre for Advanced Defence Research in Robotics and Autonomous
Systems research agreement UA216424-S25 10725.}, James Kennedy, Airlie Chapman, Peter Dower, Chris Manzie, Simon Crase}
\date{}
\begin{document}

\maketitle
\begin{abstract}
    In coverage control problems that involve time-varying density functions, the coverage control law depends on spatial integrals of the time evolution of the density function. The latter is often neglected, replaced with an upper bound or calculated as a numerical approximation of the spatial integrals involved. In this paper, we consider a special case of time-varying density functions modeled as Gaussian Mixture Models (GMMs) that evolve with time via a set of time-varying sources (with known corresponding velocities). By imposing this structure, we obtain an efficient time-varying coverage controller that fully incorporates the time evolution of the density function. We show that the induced trajectories under our control law minimise the overall coverage cost. We elicit the structure of the proposed controller and compare it with a classical time-varying coverage controller, against which we benchmark the coverage performance in simulation. Furthermore, we highlight that the computationally efficient and distributed nature of the proposed control law makes it ideal for multi-vehicle robotic applications involving time-varying coverage control problems. We employ our method in plume monitoring using a swarm of drones. In an experimental field trial we show that drones guided by the proposed controller are able to track a simulated time-varying chemical plume in a distributed manner. 
\end{abstract}
\section{Introduction}
Coverage control~\cite{cortes2004coverage_nonuniform} is a well-established distributed control technique for multi-agent robotic systems. Its objective is to distributedly deploy agents within a bounded region according to a high-level density function. The latter function is a scalar field that indicates the relative importance of different locations within the coverage region. Originally adopted from signal processing~\cite{lloyd1982least} and operations research~\cite{okabe1997locational}, coverage control has found many robotic applications including mobile sensor networks~\cite{cortes2004coverage_nonuniform}, multi-robot persistent surveillance~\cite{persistentsurveillance} and formation control~\cite{diaz2015distributed}.

 Cortes et al.~\cite{cortes2005coordination_Uniform, cortes2004coverage_nonuniform , cortes2002coverage_dynamic} were among the first to apply the classical Lloyd algorithm~\cite{lloyd1982least}, from signal processing, to the multi-robot coordination problem. In the algorithm, robots iteratively 
 tessellate the coverage space into Voronoi partitions and simultaneously track the centroid of these partitions. The tessellation and hence the control algorithm is distributed, as agents only require position information from neighbouring agents to calculate their control action.
 The uniform coverage problem was introduced in~\cite{cortes2005coordination_Uniform}. This was extended in~\cite{cortes2004coverage_nonuniform} to non-uniform coverage problems where a density function biases the distribution of the agents in the environment, spatially close to the density field. A stability proof was provided demonstrating the  asymptotic convergence of the network of robots to the centroids of the local minima of their joint coverage cost. The rate of convergence of this controller was characterised further  in~\cite{kennedy2021exponential} where it was shown that under stronger local convexity assumptions, the convergence rate can be exponential.  
 
  In~\cite{cortes2002coverage_dynamic}, the authors offered an extension to this control algorithm for a time-varying density function. However, as pointed out, e.g., in~\cite{lee2015multirobot}, this algorithm relies on a restrictive invariance condition, namely, that a spatial integral dependent on the time evolution of the density is negligible. In~\cite{kennedy2019generalized}, this assumption was relaxed and a generalised time-varying density coverage controller was introduced that relies on a worst-case bound on the rate of change of the density function. This controller can stabilise the agents to a ball around their respective centroids. A slight modification of this controller was later proposed in~\cite{abdulghafoor2021distributed}, where the same assumption was utilised to guarantee asymptotic convergence.

 Alternative coverage control techniques rely on particular time-varying density functions rather than an assumption on their rate of change. In~\cite{miah2016non}, a class of time-varying density functions are considered that obey mass conservation under diffusion. This work was extended to convection-diffusion modeled time-varying density functions in~\cite{mei2019distributed}. In~\cite{miah2017generalized}, a number of time-varying moving targets act as means of Gaussian functions, summation of which forms the density function. Asymptotic stability certificates are presented in~\cite{mei2019distributed,miah2017generalized} via bounding the associated density function derivatives.  Gaussian mixture modeling of time-varying density functions is also considered in~\cite{abdulghafoor2021distributed}. Here, initial and final Gaussian mixture models are assumed to be known and the instantaneous density is analytically computed by interpolating between the two. A coverage controller is then employed to track the instantaneous density function while also guaranteeing obstacle avoidance via modification of the Voronoi regions. The boundedness results critical to the convergence results of all these papers are reminiscent of the boundedness assumption proposed in~\cite{kennedy2019generalized}. That is, the particular density models proposed in these papers are special cases where the boundedness assumption in~\cite{kennedy2019generalized} can be verified. None of these methods provide coverage control laws that directly rely on the velocity or internal model of the density function considered. A similar idea was behind the control law  in~\cite{teruel2019distributed}, which aimed to match the velocity of a moving coverage region but not the time-varying elements of the density function itself. While most of the mentioned methods use first order derivatives of the coverage cost to derive their control law, dynamic coverage controllers  in~\cite{lee2015multirobot, diaz2015distributed} rely on second order information. These methods rely on approximations of the second-order information in order to maintain the distributed nature of their coverage control law. 

In this paper, we consider a special case of time-varying density functions modeled as Gaussian Mixture Models (GMMs) that evolve in time via a set of time-varying sources with known corresponding velocities. By imposing this structure, we provide a time-varying coverage controller that does not require further assumptions on the time-evolution of the density function and rather fully incorporates the kinematics of the density function directly. We show that this controller asymptotically minimises the coverage objective function. The GMM structure minimises the numerical spatial integrals required for the proposed controller - reducing computation.

The paper is structured as follows. Section~\ref{sec:prob} introduces the distributed multi-agent coverage control problem. Here, we also provide details of the GMM modeling of the density function and the collective coverage objective function. In Section~\ref{sec:main-results}, we present the proposed control law which we formally prove guides the agents along trajectories that minimise their collective coverage cost.  Section~\ref{sec:sims} provides a simulation study that benchmarks the performance of this control law against two well-known coverage control strategies from~\cite{cortes2004coverage_nonuniform} and~\cite{cortes2002coverage_dynamic}. Additionally, we discuss implementation issues and describe a field trial of these results. Section~\ref{sec:conc} concludes the paper.  
\section{Problem formulation}\label{sec:prob}
Consider a group of $m$ agents, positioned at $p_i\in \mathbb{R}^2$ for $i\in\mathcal{I}\triangleq\{1,\cdots,m\}$, that provides coverage of a bounded convex set $\Omega\subset\mathbb{R}^2$. Denote $P\triangleq \{p_i\;|\;i\in\mathcal{I}\}$. We are interested in the classical Voronoi tessellation~\cite{du1999centroidal} partitioning $\Omega$ into a set of polygon regions $V_i(P)$ defined as\vspace{-.25cm}
\begin{equation}\label{eq:tess}
    V_i(P) \triangleq \{q\in \Omega\;:\;\Vert q - p_i\Vert\leq \Vert q - p_j\Vert,\;  i,j\in\mathcal{I}\}.
\end{equation}
The Voronoi tessellation $V(P)\triangleq \{V_i(P)\;|\;i\in\mathcal{I}\}$ is calculated by the agents in a \emph{distributed} manner. That is, each agent $i$ calculates its own partition $V_i$ only using position information of itself and its neighbours $p_j$, $j\in\mathcal{N}_i$. The neighbour set is defined as\vspace{-2.5mm}
\begin{equation}
\mathcal{N}_i\triangleq \{j\in\mathcal{I}\; :\;  j\not=i,\; V_j(P)\cap V_i(P)\not=\emptyset \}.
\end{equation}
Note that this definition is symmetric (but not transitive), i.e.,  $j\in\mathcal{N}_i\Longleftrightarrow i\in\mathcal{N}_j$.
 The boundary between neighbouring Voronoi regions $i$ and $j$ is a set of points where 
 \begin{equation}\label{eq:boundary}
    \partial V_{ij}(P) \triangleq \{q\in V_i(P)\;:\;\Vert q - p_i\Vert= \Vert q - p_j\Vert,\; j\in \mathcal{N}_i\}.
 \end{equation}
 This is a line segment perpendicular to $p_j-p_i$ passing through the midpoint $\frac{p_j+p_i}{2}$.\vspace{-2.5mm}
 \begin{equation}\label{eq:boundary2}
     \partial V_{ij}(P) = \left\{q\in V_i:\left(q-\frac{p_j+p_i}{2}\right)^{\top}(p_j-p_i)=0\right\}. 
 \end{equation}
 We denote by $\hat{n}_{ij}$, the outward normal direction vector to $\partial V_{ij}$, that is,\vspace{-.25cm}
\begin{equation}\label{eq:normal}
    \hat{n}_{ij} \triangleq \frac{p_j -p_i}{\Vert p_j -p_i\Vert},\; \forall \; j\in\mathcal{N}_i.
\end{equation}
 The overall boundary of partition $V_i$ is the union set of its neighbouring boundaries $\{\partial V_{ij}\}$ and the intersection of $V_i$ with the coverage region boundary $\partial \Omega$,
 \begin{equation}
     \partial V_i(P)\triangleq  \bigcup_{j\in\mathcal{N}_i} \partial V_{ij}(P) \cup (V_i(P) \cap \partial\Omega).
 \end{equation}

 In many coverage control problems, agents are to cover the coverage region in a non-uniform manner. Let $\phi:$ $\mathbb{R}^2\times \mathbb{R}$ $\longrightarrow \mathbb{R}^+$ denote a time-varying density function that is given to all agents. This function indicates the relative importance of the coverage task at any particular location $q\in\Omega$ and time $t$. We  introduce a number of auxiliary variables that appear in many coverage control algorithms.    

Let $m_i(t)\in\mathbb{R}$ and $c_i(t)\in\mathbb{R}^2$ denote the mass and centroid of partition $V_i$ distributed according to $\phi(q,t)$, defined by\footnote{When clear from the context and for brevity we suppress state and time dependence. Nevertheless, partitions $V_i(P)$ depend on time as they depend on agent positions which change in time. Thus, mass and centroid also depend on state $P$ and time as they are defined with respect to $V_i(P)$ and $\phi(q,t)$.}
\begin{equation}\label{eq:masscentroid}
    m_i\triangleq \int_{V_i} \phi(q,t) dq, \;c_i\triangleq \frac{1}{m_i}\int_{V_i} q\phi(q,t) dq.
\end{equation}
The partial derivative of these entities with respect to time is     
\begin{equation}\label{eq:partialmasscentroid}
\begin{split}
    &\frac{\partial m_i}{\partial t} =\int_{V_i} \frac{\partial \phi}{\partial t}(q,t) dq,\\ 
    &\frac{\partial c_i}{\partial t} =\frac{1}{m_i}\bigg(\int_{V_i} q\frac{\partial \phi}{\partial t}(q,t) dq - \frac{\partial m_i}{\partial t}c_i\bigg).
\end{split}
\end{equation}
In both equations, we used the fact that the partial time derivative of the partition $V_i$ is zero as it is only indirectly dependent on time through $p_i(t)$ and $\{p_j(t)\}_{j\in\mathcal{N}_i}$. The second equation is obtained by applying the chain rule to $\frac{\partial (c_im_i)}{\partial t}$.

\subsection{Time-Varying Gaussian Mixture Modeled Density}\label{sec:GMM}
 We model the density function using a Gaussian Mixture Model (GMM). It is well established that  GMMs can approximate a continuous nonlinear function arbitrarily well given a sufficient number of radial basis functions~\cite{poggio1989theory}. The mathematical structure of GMMs is intuitive and facilitates human-driven design. For instance, one can design the means, standard deviations, and weights of a GMM to generate various swarm formations~\cite{diaz2015distributed}. GMMs have also been successfully deployed in the context of estimation problems~\cite{schwager2009decentralized, leong2022field, leong2022logistic} where the density function is estimated or learned as a GMM structure. 

Consider the following GMM-structured density function $\phi$ with $K\in\mathbb{R}^+$ Gaussian density functions each with a known time-varying mean $s_k(t)\in\mathbb{R}^2$, time-invariant mixture weight $a_k\in\mathbb{R}$ and standard deviation $\sigma_k\in\mathbb{R}^+$, where $k\in\{1,\cdots,K\}$. The means $s_k(t)$ are known for all $t$, that is, their time evolution is predetermined with initial values $s_{k,0}\in\mathbb{R}^2$ and their velocities $w_k(t)\in\mathbb{R}^2$ for all $t$.
\begin{equation}\label{eq:GMM}
    \begin{split}
        & \phi(q,t)\triangleq \sum_{k=1}^K\phi_k(q,t),\\
        & \phi_k(q,t) \triangleq  a_k\exp\left({\frac{-\Vert q - s_k(t)\Vert ^2}{2\sigma_k^2}}\right),\\
        & \dot{s}_k(t) = w_k(t),\;s_k(0)=s_{k,0}.
    \end{split}
\end{equation}
By applying the chain rule, differentiation of~\eqref{eq:GMM} yields  
\begin{equation}\label{eq:grad-dot}
    \begin{split}
        & \frac{\partial{\phi}}{\partial t}(q,t) =  \sum_{k=1}^K\frac{w_k(t)^{\top}( q - s_k(t))}{\sigma_k^2}\phi_k(q,t),\\
        & \nabla \phi_k(q,t) = \frac{ -(q - s_k(t))}{\sigma_k^2}\phi_k(q,t).\\
    \end{split}
\end{equation}    
Substituting one into the other yields the partial differential equation
    \begin{equation}\label{eq:PDE}
        \frac{\partial{\phi}}{\partial t}(q,t) = - \sum_{k=1}^K w_k(t)^{\top}\nabla \phi_k(q,t).
    \end{equation}
The mass $m_{ik}$ and centroid $c_{ik}$ of the $k$th Gaussian component are defined as 
\begin{equation}\label{eq:componentmasscenter}
    \begin{split}
        &m_i= \sum_{k=1}^K m_{ik}, \;m_{ik} \triangleq \int_{V_i} \phi_k(q,t) dq,\\
        &c_i=\frac{\sum_{k=1}^K m_{ik}c_{ik}}{\sum_{k=1}^K m_{ik}}, \; c_{ik} \triangleq \frac{1}{m_{ik}}\int_{V_i} q\phi_k(q,t) dq.
    \end{split}
\end{equation}

\subsection{Coverage Control Problem}
The collective coverage objective is given by 
\begin{equation}\label{eq:macost}
   H(P,t) \triangleq \frac{1}{2}\int_{\Omega} \min_i \Vert q - p_i\Vert ^2\phi(q,t) dq.
\end{equation}
Minimising this cost function with respect to agent positions $P$ leads to dispersing the agents over the coverage region $\Omega$ such that the integrated distance of points $q\in\Omega$ with respect to their closest agent, weighted by the density value of $q$, is minimised. 

By~\eqref{eq:tess}, the cost~\eqref{eq:macost} can be equivalently expressed as a finite summation over the Voronoi tessellation, as shown in~\cite{du1999centroidal}
\begin{equation}\label{eq:macoste}
   H_V(P,t) = \frac{1}{2}\sum_i \int_{V_i}\Vert q - p_i\Vert ^2\phi(q,t) dq.
\end{equation}
This formulation is amenable to a distributed coverage control problem,  
 as the collective coverage cost is distributed as a summation of costs, each associated with a particular agent. Before we formally introduce the coverage control problem, we define agent dynamics
\begin{equation}\label{eq:dist_cont}
    \dot{p}_i = f(p_i, \{p_j\}_{ j\in\mathcal{N}_i},t),
\end{equation}
where $f(p_i, \{p_j\}_{ j\in\mathcal{N}_i},t)$ denotes a distributed control law. 
A control law is distributed if it only relies on information locally available to each agent. We also assume that the coverage region $\Omega$ and the density function $\phi$ are available to all agents whereas each agent has access only to its own position and the positions of its neighbors.  
\begin{problem}[Dynamic Coverage]\label{prob:1}
    Consider a time-varying density function with the GMM structure~\eqref{eq:GMM} that is fully specified with all of its static and time-varying parameters. Find a distributed control strategy of the form~\eqref{eq:dist_cont} that minimises the collective coverage cost~\eqref{eq:macost}.
\end{problem}

\section{Main Results}\label{sec:main-results}
In this section, we propose a control law that solves Problem~\ref{prob:1}. The proposed controller is obtained by exploiting the GMM modeling introduced in~\ref{sec:GMM}. It incorporates structural elements of the GMM introduced in~\eqref{eq:GMM}, in particular, the Gaussian source velocities $w_k$. 

Consider the control law 
\begin{equation}\label{eq:full-controller}
\begin{split}
    &f(p_i, \{p_j\}_{ j\in\mathcal{N}_i}) \triangleq \frac{\sum_k m_{ik}w_k}{m_{i}} \\&\quad\quad -\frac{1}{2}\left( \beta - \frac{F_{i}}{m_{i}\Vert p_i-c_{i}\Vert ^2}\right)(p_i-c_{i})  ,\\
    & F_{i}\triangleq  \sum_{k=0}^K 2m_{ik}w_k^{\top}(c_i - c_{ik}) \\
    &\quad\quad  +\sum_{k=0}^K \int_{\partial V_i\cap \partial \Omega} \Vert q - p_i\Vert ^2 w_k^{\top}\hat{n}_i\phi_k(q,t)dq .
    \end{split}
\end{equation}
The first term is a weighted average of Gaussian component velocities $w_k$. Each $w_k$ is scaled according to the proportion of the total mass in partition $V_i$ that originates from the Gaussian component $k$. This prompts the velocity of agent $i$ to partly match the average Gaussian source velocities, weighted more towards Gaussian components that have a larger mass inside its partition $V_i$. The second part involves a proportional feedback term (with feedback gain $\beta>0$) that guides the agent towards its partition centroid $c_i$. The scalar $F_{i}$ is a gain correction term that compensates further for dynamics of the density function. The vector $\hat{n}_i$ is a unit norm vector perpendicular to the boundary line segments in $\partial V_i\cap \partial \Omega$.

 Importantly, the integral term in $F_{i}$ only needs to be computed if the partition boundary $\partial V_i$ has overlap with the boundary of the overall coverage region $\partial \Omega$.\footnote{If we further assume the GMM function vanishes over $\partial \Omega$, this integral term is always zero. In many applications, the density function is designed to shape the behaviour of agents in a particular way. Therefore, by design, this assumption can be enforced.} In the worst case scenario that this is not possible, this term requires computations of a few boundary integrals. This is in contrast with the spatial integrals that would be required otherwise.

 Next we contrast this controller with the one proposed in~\cite{cortes2002coverage_dynamic}.\vspace{-.25cm}
\begin{equation}\label{eq:dynamic-cortes}
\begin{split}
    &f_{dl}(p_i, \{p_j\}_{ j\in\mathcal{N}_i}) = \frac{\partial c_i}{\partial t}  -\frac{1}{2}\left( \frac{1}{m_i}\frac{\partial m_i}{\partial t} + \beta\right)(p_i-c_{i}).
    \end{split}
\end{equation}

Both controllers are fully distributed and include $\beta>0$ as a proportional centroid-tracking gain. In this controller however, the time-varying aspects of the density function are only incorporated, approximately and implicitly, via the auxiliary variables $\frac{\partial m_i}{\partial t}$ and $\frac{\partial c_i}{\partial t}$ which were introduced in~\eqref{eq:partialmasscentroid}.

The following theorem states that the same cost decrease rate in~\cite{cortes2004coverage_nonuniform} and~\cite{cortes2002coverage_dynamic} can be shown for control law~\eqref{eq:full-controller}.
\begin{theorem}\label{th:decrease}
    For $\beta > 0$,  under the GMM assumptions~\eqref{eq:GMM},  control law~\eqref{eq:full-controller}   
solves Problem~\ref{prob:1} by asymptotically minimising the collective coverage cost~\eqref{eq:macoste}. Moreover, the closed-loop trajectories of~\eqref{eq:dist_cont}, under control law~\eqref{eq:full-controller} result in the same cost decrease condition shown in~\cite{cortes2004coverage_nonuniform}, namely \vspace{-.4cm}
\begin{equation}\label{eq:dec_cond}
  \frac{d}{dt} H_V(P,t) =  -\frac{\beta m_i}{2}\sum_i  \Vert p_i-c_{i}\Vert ^2.
\end{equation}
\end{theorem}

    
    
 

\begin{proof}
Consider the total time derivative of~\eqref{eq:macoste}. 
\begin{equation}\label{eq:mscost-dot}
    \begin{split}
        &\frac{d}{dt} H_V(P,t) = \sum_i \bigg(\frac{1}{2}\int_{V_i}\Vert q - p_i\Vert ^2\frac{\partial{\phi}}{\partial t}(q,t)dq \\
        & \quad -\int_{V_i}  \dot{p}_i^{\top}(q-p_i)\phi(q,t) dq \\
        &\quad +\frac{1}{2}\int_{\partial V_i}\Vert q - p_i\Vert ^2\phi(q,t)\hat{n}_i^{\top}\frac{d (\partial V_{i})}{dt}dq \\
        &\quad -\sum_{j\in\mathcal{N}_i} \frac{1}{2}\int_{\partial V_{ij}}\Vert q - p_j\Vert ^2\phi(q,t)\hat{n}_{j}^{\top}\frac{d (\partial V_{ij})}{d t}dq\bigg).\\
        \end{split}
\end{equation}
The last two terms are due to the Leibniz integral rule and the dependency of partition $V_i$ and neighbouring partitions $\{V_{ij}\}_{j\in\mathcal{N}_i}$ on $p_i$ and therefore indirectly on time $t$. These two terms cancel each other out since over each boundary $\partial V_{ij}$, from definition~\eqref{eq:boundary}, we have $\Vert q - p_i\Vert = $ $\Vert q - p_j\Vert$ while $\hat{n}_{ij}=-\hat{n}_{ji}$ from~\eqref{eq:boundary2}. We have also used the fact that $\frac{\partial \Omega}{\partial t}=0$.

Next, substitute the PDE~\eqref{eq:PDE} into the first term and note that the second integral simplifies from  definitions~\eqref{eq:masscentroid}.
\begin{equation}
\begin{aligned}\label{eq:mscost-dot2}
        &\frac{dH_V}{dt} (P,t) =\sum_i \left(-\frac{1}{2}\int_{V_i}\Vert q - p_i\Vert ^2 \sum_k w_k^{\top}\nabla \phi_k(q,t)dq \right.\\
        &\quad +m_i(p_i - c_i)^{\top}\dot{p}_i\bigg).
    \end{aligned}
\end{equation}
Recall Green's first identity for integration by parts. Let $u:\mathbb{R}^2\longrightarrow \mathbb{R}^2$ and $f:\mathbb{R}^2\longrightarrow \mathbb{R}$, then
    \begin{equation}\label{eq:ibp}
    \begin{split}
    & \int_{V} u(q)^{\top} \nabla f(q)dq = \int_{\partial V} u(q)^{\top} \hat{n}f(q) dq\\&\quad\quad - \int_{V} (\nabla \cdot u(q)) f(q) dq,\\   
    \end{split}
\end{equation}
where $\hat{n}$ is the outward unit normal vector to the boundary $\partial V$ and $\nabla \cdot u(q)$ is the divergence of the vector field $u$.
Now apply integration by parts~\eqref{eq:ibp} to the first integral, considering $u(q) \leftarrow \Vert q - p_i\Vert ^2  w_k$ and $f(q) \leftarrow \phi_k(q,t)$ and noting that $\nabla\cdot( \Vert q - p_i\Vert ^2  w_k) = 2w_k^{\top}(q-p_i)$.  
\begin{equation}\label{eq:aux1}
    \begin{split}
    & \sum_i \big(-\frac{1}{2}\int_{V_i}\Vert q - p_i\Vert ^2 \sum_k w_k^{\top}\nabla \phi_k(q,t)dq\big) \\& \quad =\sum_i \bigg(- \frac{1}{2}\int_{\partial V_i}\Vert q - p_i\Vert ^2 \sum_k w_k^{\top}\hat{n}_i \phi_k(q,t)dq\\
    &\quad\quad\quad +\sum_k w_k^{\top}\int_{ V_i} (q - p_i)\phi_k(q,t)dq\bigg).
    \end{split}    
\end{equation}
Here the boundary integral over $\partial V_i$ entails a number of line integrals over $\partial V_{ij}$ for all $j\in\mathcal{N}_i$ and possibly\footnote{That is when the Voronoi region $V_i$ shares a boundary with the boundary of the overall coverage space $\partial \Omega$. } $\partial \Omega\cap \partial V_i$ with $\hat{n}_i$ specialising to outer normal vectors to neighbour boundaries $\hat{n}_{ij}$~\eqref{eq:normal} or to outer normal vectors to coverage space boundary $\Omega$, respectively. 

 The first integral in~\eqref{eq:aux1} simplifies to an integral over $\partial \Omega\cap \partial V_i$ (if non-empty) since the integrals over inter-agent boundaries cancel out. The latter is because by definition~\eqref{eq:boundary}, over inter-agent boundaries we have $\Vert q-p_i\Vert = \Vert q-p_j\Vert$ and that $\hat{n}_{ij}=-\hat{n}_{ji}$. The second integral in~\eqref{eq:aux1}, after adding and subtracting $c_i$ and using~\eqref{eq:componentmasscenter}, yields
\begin{equation}
\begin{split}
    &\sum_k w_k^{\top}\int_{ V_i} (q - p_i +c_i-c_i)\phi_k(q,t)dq \\
    &\quad = \sum_k m_{ik}w_k^{\top}(c_{ik}-c_{i})-(p_i-c_i)^{\top}\sum_k m_{ik}w_k.
    \end{split}
\end{equation}
Thus, we simplify~\eqref{eq:aux1} to
\begin{equation}\label{eq:mscost-dot3}
\begin{split}
&\frac{d}{dt} H_V(P,t) =\sum_i \bigg(-(p_i-c_{i})^{\top}\sum_k m_{ik}w_k \\&\quad\quad -\frac{1}{2}\sum_k \int_{\partial V_i\cap \partial \Omega}\Vert q - p_i\Vert ^2w_k^{\top}\hat{n}_j\phi_k(q,t)dq \\ &\quad \quad+  \sum_k m_{ik}w_k^{\top}(c_{ik}-c_{i}) +m_{i}(p_i - c_{i})^{\top}\dot{p}_i\bigg).
    \end{split}
\end{equation}

Substituting~\eqref{eq:dist_cont} and control law~\eqref{eq:full-controller} into~\eqref{eq:mscost-dot3}, the decrease condition~\eqref{eq:dec_cond} is reached.

\end{proof}


\vspace{-.5cm}
\section{Simulation and Field Trial}\label{sec:simsfiled}
In this section we first discuss some considerations on how to implement the proposed algorithm in Section~\ref{sec:implementation}. Section~\ref{sec:sims} describe a simulation study that involves comparing the proposed algorithm against the continuous Lloyd algorithm~\cite{cortes2004coverage_nonuniform} (serving as a baseline since it was not designed for dynamic coverage control) and its related dynamic coverage control extension introduced in~\cite{cortes2002coverage_dynamic}. Finally, we will describe in Section~\ref{sec:outdoor} an outdoor experiment conducted on a number of quadrotor aerial platforms tracking a simulated moving plume serving as the dynamic density function. 

\subsection{Implementation}\label{sec:implementation}
One practical aspect of the proposed control law is around robust implementation of terms that are divided by $\Vert p_i-c_i\Vert^2$. If $p_i=c_i$ this control law ceases to be well defined. To tackle this,  
when it is detected that $\Vert p_i-c_i\Vert\leq \epsilon$ with an arbitrarily small threshold $\epsilon>0$, we modify the controller by excluding $F_i$.
\begin{equation}\label{eq:switching}\small
   \begin{split}
       &\Vert p_i-c_i\Vert\leq \epsilon\;\Longrightarrow\\ &\quad f(p_i, \{p_j\}_{ j\in\mathcal{N}_i}) \longleftarrow \frac{\sum_k m_{ik}w_k}{m_{i}} -\frac{\beta}{2}(p_i-c_{i}).
   \end{split} 
\end{equation}
Note that when the proximity assumption of~\eqref{eq:switching} holds, the same cost decrease rate~\eqref{eq:dec_cond} can be shown. This can be seen from the analysis provided in the proof of Theorem~\ref{th:decrease}. Nevertheless, a formal analysis of the overall hybrid controller is beyond the scope of this paper.  

 Furthermore, we introduce a general clamping strategy that prevents the agents from exceeding their top speed limit $s_m\in\mathbb{R}^+$.
 \begin{equation}
 \dot{p}_i = \left\{\begin{matrix}
     f(p_i, \{p_j\}_{ j\in\mathcal{N}_i}), & \mbox{if }\Vert f(p_i, \{p_j\}_{ j\in\mathcal{N}_i})\Vert \leq s_m,\\
     \frac{s_mf(p_i, \{p_j\}_{ j\in\mathcal{N}_i})}{\Vert f(p_i, \{p_j\}_{ j\in\mathcal{N}_i})\Vert}, & \hspace{-2mm}\mbox{otherwise.}
 \end{matrix} \right.\vspace{-5mm}
 \end{equation}
\subsection{Simulation Study}\label{sec:sims}
In this section we consider a simulated plume coverage scenario in a bounded rectangular area of $100\times 200$m size (bottom left corner located at the origin). Consider five agents at initial x-y coordinates at columns of 
\begin{equation}\small
P=\begin{bmatrix}
    5 & 5 & 5 & 5 & 5\\
    5 & 25 & 45 & 65 & 85
\end{bmatrix}.    
\end{equation}
All agents have a top speed of $s_m=3.5\; (m/s)$ and a uniform control gain of $\beta=0.05$.
The plume is a five-component GMM with the following parameters. The GMM sources, arranged as columns of $S_l\in\mathbb{R}^{2\times 5}$, move linearly between six different configurations  ($l\in\{0,\cdots,5\}$), 
\begin{equation}\small
\begin{split}
    &S_0=\begin{bmatrix}
        30 & 55 & 85 & 100 &110\\
        15&25&55&22&35
    \end{bmatrix},\\
    &S_1=\begin{bmatrix}
        50 & 65 & 90 & 110 &120\\
        20&35&65&42&45
    \end{bmatrix},\\
    &S_2=\begin{bmatrix}
        70 & 95 & 105 & 110 &130\\
        40&45&65&62&55
    \end{bmatrix},\\
    &S_3=\begin{bmatrix}
        90 & 115 & 125 & 130 &145\\
        20&55&75&62&65
    \end{bmatrix},\\
    &S_4=\begin{bmatrix}
        110 & 130 & 135 & 140 &150\\
        30&60&80&85&75
    \end{bmatrix},\\
    &S_5=\begin{bmatrix}
        140 & 145 & 150 & 150 &175\\
        35&70&92&95&60
    \end{bmatrix}.\\
\end{split}
\end{equation}
\begin{figure*}[h!]
\centering
\includegraphics[width=.9\linewidth]{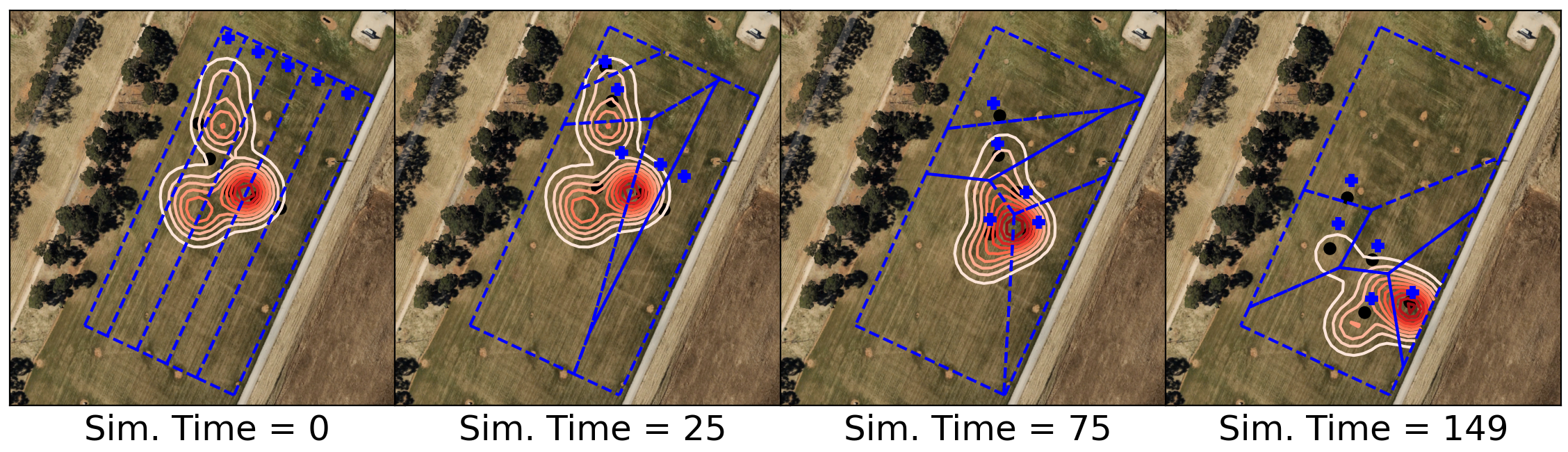}\\
\caption{Configuration of the plume density and agents during the simulation. Dashed blue lines represent the Voronoi tessellation. Blue crosses represent the agents. Contour lines of the GMM density are shown with shades of red. Black dots represent the centroids of the Voronoi regions. The second frame is at the simulation time when the GMM density starts moving.}
    \label{fig:t1}
\end{figure*}
  We let $\sigma_k=15, a_k=100$ for each Gaussian component $k$.  Figure~\ref{fig:t1} depicts the GMM configuration and its evolution within the coverage region considered (which is based on our outdoor experiment area in Section~\ref{sec:outdoor}). Note that the GMM moves significantly faster than the top speed of the agents. This is a deliberate choice designed to challenge the control methods against a fast moving density. If agent speeds are unlimited or if aggressive centroid tracking is used with a high proportional gain $\beta$, then the performance of these methods will be nearly identical. In practice, high control gains are undesirable as they lead to jerky maneuvers and energy waste.  
  
   We compare the proposed control law~\eqref{eq:full-controller}, dubbed `GMM Controller' against the continuous Lloyd algorithm~\cite{cortes2004coverage_nonuniform}, dubbed `Lloyd', and its dynamic coverage control extension introduced in~\cite{cortes2002coverage_dynamic}, dubbed `Dynamic Lloyd'. Note that these labels are for the sake of this study and do not represent official names.
  
\begin{figure}[h]
    \centering   \includegraphics[width=.9\linewidth]{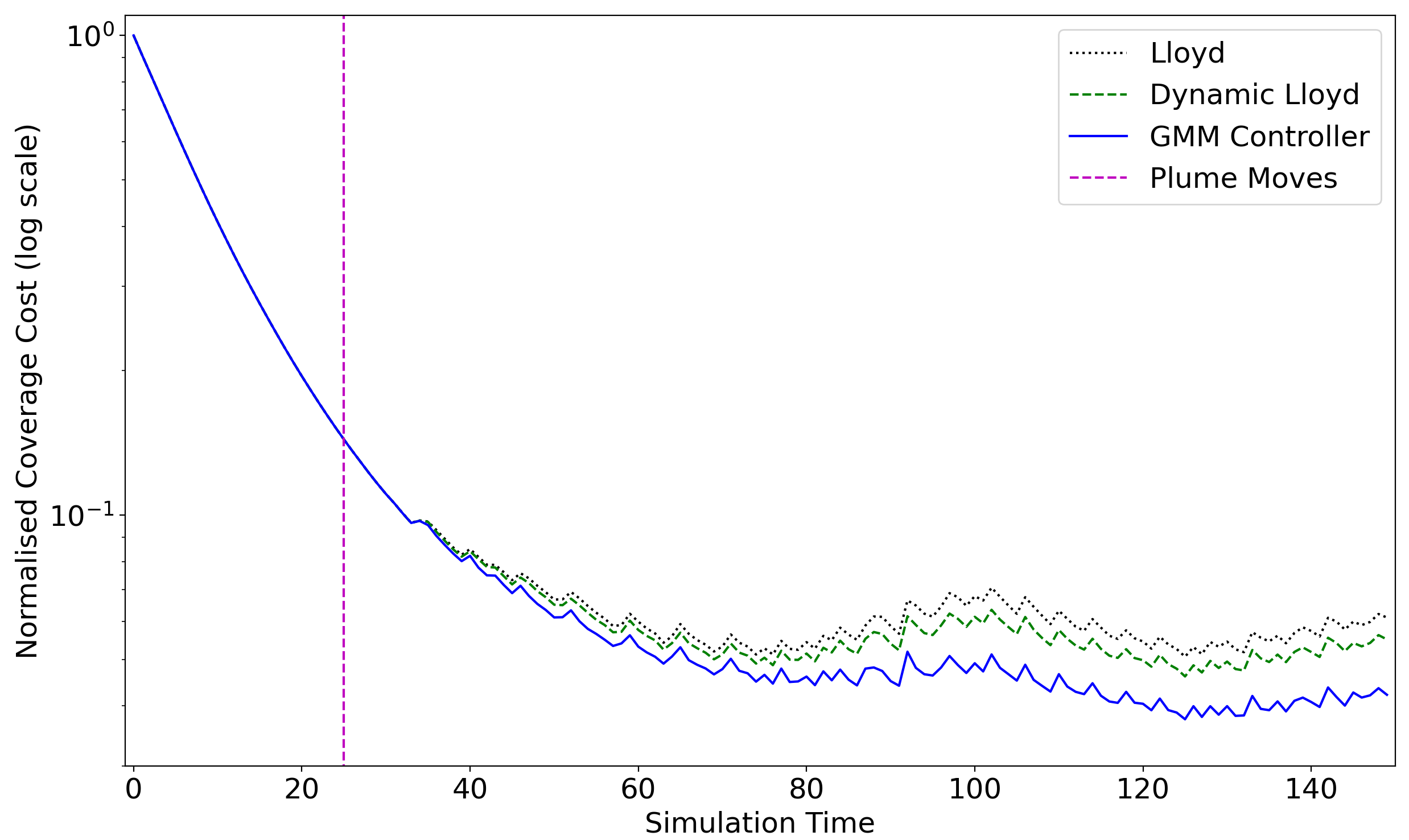}
    \caption{Coverage cost $H(V,P)$~\eqref{eq:macost} of methods vs simulation time. Note that the coverage cost values (on the logarithmic y-axis) are normalised with respect to the coverage cost at the start of the simulation.}
    \label{fig:CovCost}
\end{figure}
 Figure~\ref{fig:CovCost} illustrates the coverage cost~\eqref{eq:macost} of each method during this scenario.  As can be seen, $H(P,t)$ is minimised as expected by Theorem~\ref{th:decrease}. The proposed GMM Controller outperforms the rest of the methods considered, as it fully incorporates the time-varying dynamics of the GMM density into its control action. When the plume starts moving (as marked by the vertical line in Figure~\ref{fig:t1}) the proposed controller outperforms the dynamic controller~\eqref{eq:dynamic-cortes}, which itself outperforms the non-dynamic Lloyd algorithm. Before this mark, the performances of all methods are dominated by their centroid tracking action. 

\subsection{Outdoor Experiments}\label{sec:outdoor}
Now we describe a similar experiment to that of Section~\ref{sec:sims}, this time performed on real hardware in an outdoor experiment to further verify our results. Figure~\ref{fig:t1} presents a satellite view of our outdoor experiment site. The proposed controller~\eqref{eq:full-controller} was implemented on a NVIDIA Jetson Nano acting as a companion computer onboard of four MR4 drones (see Figure~\ref{fig:bask-silvus}). The companion computer received the GPS position of its own drone from the Ardupilot Cube Orange autopilot module of MR4 and position of neighbouring drones through the SL4200 SilvusTechnologies MIMO radio network (see Figure~\ref{fig:bask-silvus}). The Silvus radios are able to establish a mesh network, though there are no restrictions on the network topology for the experiment. The parameters of the experiment, including number of agents and close operating proximity, results in a fully-connected network topology. The network inputs were received at a frequency of $10$ Hz. The full trajectory of the dynamic density function, as described in the previous section, is uploaded to each drone offline. During the experiment, the companion computer executes~\eqref{eq:dist_cont} and~\eqref{eq:full-controller}, with the adjustments discussed in Section~\ref{sec:implementation},  incorporating a control barrier function algorithm~\cite{ames2016control} for collision safety. It then issues waypoint commands at $10$ Hz to the Ardupilot Cube Orange autopilot module. Figure~\ref{fig:outdoor} shows the coverage cost reduction achieved during the experiment. As shown in Figure~\ref{fig:outdoor}, the performance closely matches the simulation results in Figure~\ref{fig:CovCost}.      
\begin{figure}[h]
    \centering
    \includegraphics[width=0.3\linewidth]{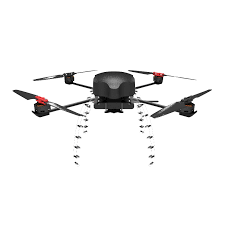}
    \includegraphics[width=0.3\linewidth]{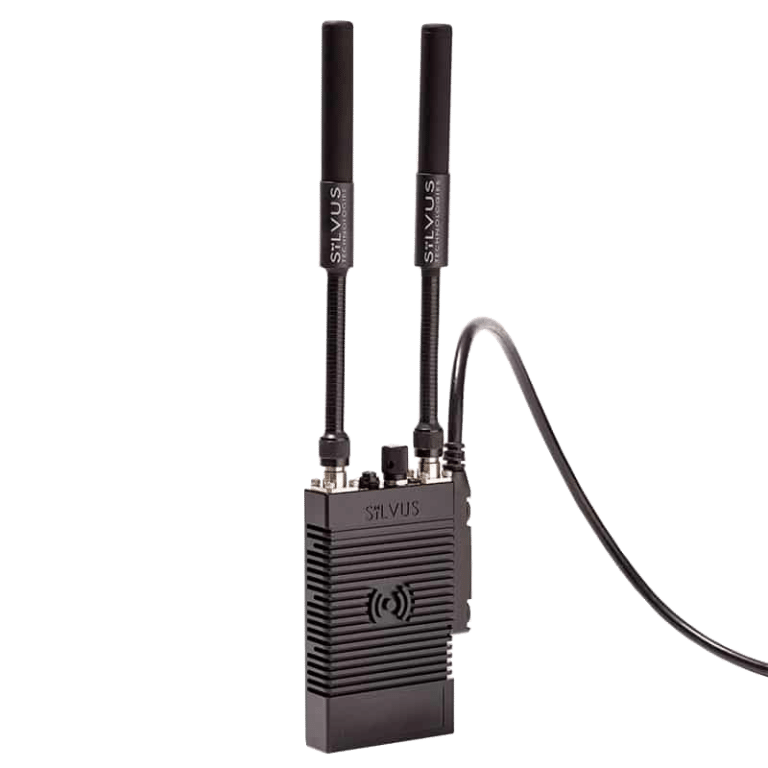}
    \caption{AeroDrone MR4 drone from BaskAerospace (left picture) equipped with SL4200 MIMO radio from SilvusTechnologies (right picture).}\vspace{-.5cm}
    \label{fig:bask-silvus}
\end{figure}
\begin{figure}[h]
    \centering
    \includegraphics[width=.8\linewidth]{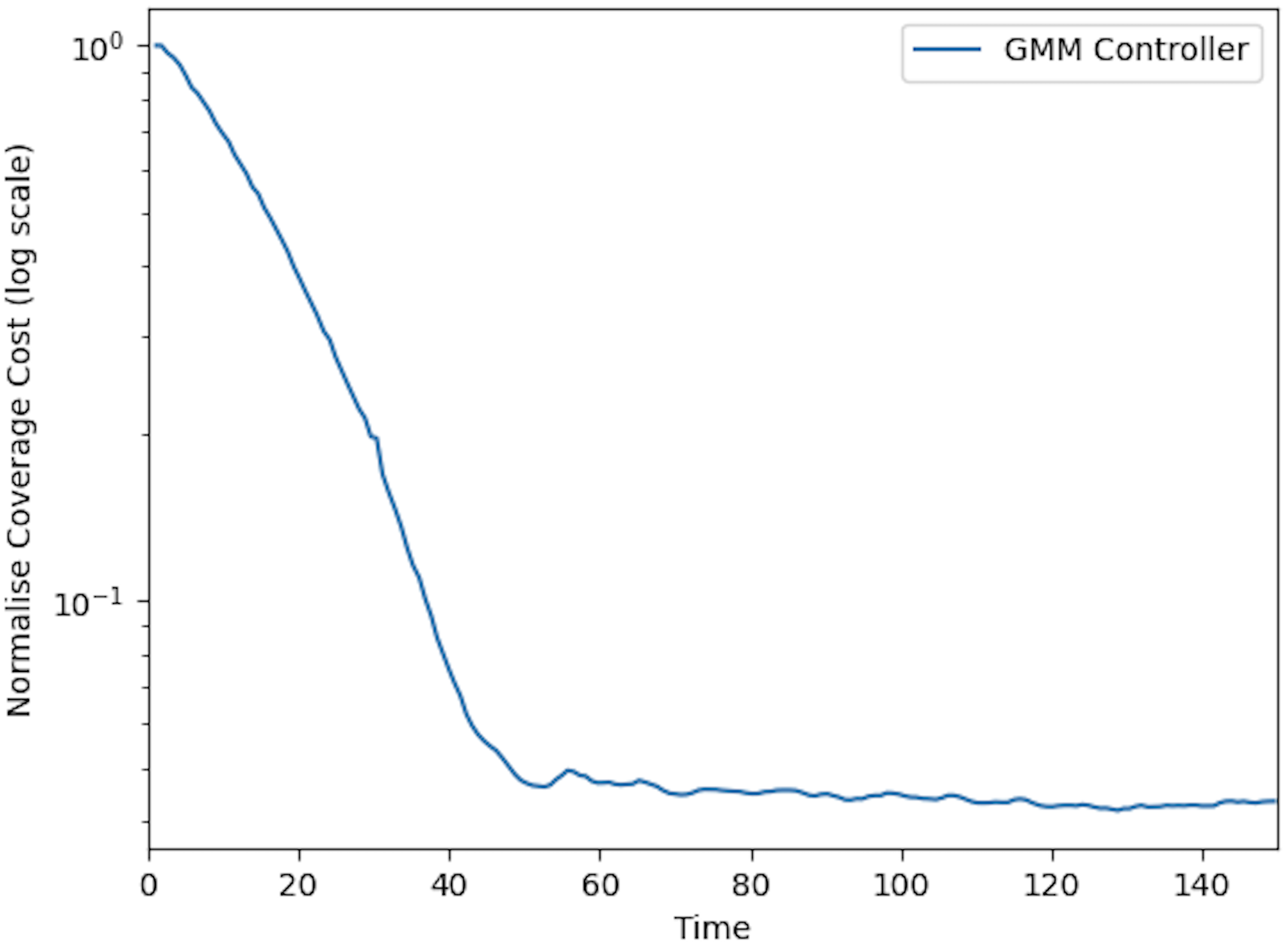}
    \caption{Coverage Cost of the proposed GMM controller during the outdoor experiment with four MR4 drones.\vspace{-1cm}}
    \label{fig:outdoor}
\end{figure}
\section{Conclusion}\label{sec:conc}
We have developed a coverage control algorithm for applications involving dynamic density functions modeled as Gaussian mixtures with known Gaussian source kinematics. This GMM structure allows us to fully incorporate and efficiently calculate the control inputs dependent on the time evolution of the density. We prove that the proposed control law generates agent trajectories along which the collective coverage cost is minimised. The computational efficiency makes the algorithm particularly attractive for practical scenarios, as demonstrated in our field trial.  

\bibliographystyle{IEEEtran}
\bibliography{ref}
\end{document}